\documentclass[twocolumn,aps,pra]{revtex4-1}
\usepackage[usenames,dvipsnames]{xcolor} 
\usepackage{ulem} 

\usepackage{lineno}
\usepackage{subfigure}
\usepackage{sublabel}
\usepackage{dcolumn}
\usepackage{amsmath,amssymb}
\usepackage{bm}
\usepackage{bbm}
\usepackage{color}
\usepackage{overpic}
\usepackage{latexsym}
\usepackage{epstopdf}
\usepackage{color}
\usepackage[english]{babel}
\usepackage{times}
\usepackage{latexsym}
\usepackage{stmaryrd}
\usepackage{psfrag,graphicx}
\usepackage{epsf}
\usepackage{subfigure}
\usepackage{amsmath}
\usepackage{amssymb}
\usepackage{amsfonts}
\usepackage{natbib}
\usepackage{epstopdf}
\usepackage{appendix}
\usepackage{enumitem}

\definecolor{mygrey}{gray}{0.35}
\definecolor{myblue}{rgb}{0.2,0.2,0.8}
\definecolor{myzard}{cmyk}{0,0,0.05,0}
\definecolor{mywhite}{rgb}{1,1,1}
\definecolor{myred}{rgb}{0.9,0.1,0.}
\usepackage[colorlinks=true,citecolor=myblue,linkcolor=myblue,urlcolor=myblue]{hyperref}

\newtheorem{theorem}{Theorem}

\newtheorem{definition}{Definition}
\newtheorem{corollary}{Corollary}

\newtheorem{proposition}{Proposition}
\newenvironment{proof}{\noindent\textbf{Proof:}}{\hfill$\blacksquare$\newline}

\newcommand{\diag}{\operatorname{\bf{diag}}} 
\newcommand{\ket}[1]{\vert #1 \rangle} 
\newcommand{\bra}[1]{\langle #1 \vert} 

\begin{document}

\title{A note on coherence power of N-dimensional unitary operators}

\author{M. Garc\'ia-D\'iaz, D. Egloff and M.B. Plenio}
\affiliation{Institut f\"{u}r Theoretische Physik, Albert-Einstein-Allee 11,
Universit\"{a}t Ulm, 89069 Ulm, Germany}

\date{\today}

\begin{abstract}
The coherence power of a quantum channel, that is, its ability to increase the coherence
of input states, is a fundamental concept within the framework of the resource theory of
coherence. In this note we discuss various possible definitions of coherence power. Then
we prove that the coherence power of a unitary operator acting on a qubit, computed with
respect to the $l_1$-coherence measure, can be calculated by maximizing its coherence gain
over pure incoherent states. We proceed to show that this result fails for dimensions $N>2$,
that is, the maximal coherence gain is found when acting on a state with non-vanishing
coherence.
\end{abstract}
\maketitle

\section{Introduction}
The development of quantum information science has led to a reassessment of quantum physical
properties such as non-locality or entanglement, elevating them to {\it resources} that may
be exploited to achieve tasks that are impossible when these properties are not available.
The quantitative theory of entanglement \cite{vedral1998entanglement,plenio2007introduction}
was perhaps the first example of a theory that was formulated by taking seriously the idea
that quantum properties are physical resources. The starting point was to take the view that
constraints, here the restriction to local operations and classical communication, prevent
certain non-local physical operations from being realizable unless resources, here entangled
states, are available which may be consumed to allow us to overcome the imposed constraints
\cite{BrandaoP08,BrandaoP10}. This viewpoint has proven fruitful as an impetus for theory to
establish a unified and rigorously defined framework for a quantitative theory of physical
resources by addressing the three principal issues: (i) the characterization, (ii) the
quantification and (iii) the manipulation of quantum states under the imposed constraints.
This framework is being explored for entanglement \cite{vedral1998entanglement,plenio2007introduction},
specific formulations of quantum thermodynamics \cite{BrandaoHO+11,GourMN+13} and of reference
frames \cite{GourS08,GourMS2009} and has led to the recognition of deep interrelations between
the theories of entanglement and the second law \cite{BrandaoP08,BrandaoP10}.

Recently, \cite{BaumgratzCP14} formulated a resource theory for quantum coherence, which is a
fundamental trait of quantum mechanics. In this work the authors defined a number of coherence
measures and outlined, following the example of the theory of entanglement, various extensions
that would have to be completed to explore all the aspects of the resource theory of coherence.
This includes the study of the interconversion of coherent states by means of incoherent operations
both, in the single copy \cite{DuBG2015,DuBQ2015,QiBD2015} and the asymptotic regime \cite{WinterY2015}
as well as the characterisation of incoherent operations \cite{StreltsovRB+2015,ChitambarH2015}.
Although not addressed from the perspective of resource theory, \cite{Aberg2006,LeviM2013} have
also dealt with the quantification of quantum coherence and the formal characterization of
coherence-decreasing processes. The relationship between coherence and entanglement has been
studied from various angles \cite{KilloranSP2015,StreltsovSD+2015,YaoXG+2015}.

Aside of these developments it was pointed out in \cite{BaumgratzCP14} that following the example
of entanglement theory \cite{EisertJP+2000,ZanardiZF00} it would be natural to develop a quantitative
theory of the coherence of operations which may have applications in the study of coherence in
dynamical processes including biological systems where the presence and role of coherence remains
a matter of current debate \cite{HuelgaP2013,CarusoCD+2010}. Indeed, first steps in
this direction were taken in \cite{ManiK2015,BuZW2015} which mostly considered the coherence power of
operations when acting on incoherent states. In our work we will demonstrate that while being
consistent, this is too restrictive as it can be shown that the achievable coherence gain can
be higher when accepting states as input which already possess some coherence~\cite{Note_added}. This mirrors
similar observations in the realm of entanglement theory \cite{BennettHL+2003,LeiferHL2003}.

After this introduction, in section \ref{Section2} of our manuscript we repeat some
basic definitions concerning coherence measures which will be followed by a discussion
of possible definitions of coherence properties of operations. This will be followed in
section \ref{Section3} by a discussion of the coherence power of unitaries on qubits which
we prove to be achieved on incoherent states. Section \ref{Section4} then proceeds to
demonstrate by means of two simple examples that for higher dimensional systems the largest
gain in coherence is typically achieved on states with coherence. We conclude with a
summary and outlook.

\section{Basic Definitions}\label{Section2}
In this section we provide the basic definitions of the quantities that we will be
exploring in this work.

{\it Measures of coherence of states --} One result of the resource theory of coherence
are well-defined quantifiers of coherence, coherence measures, which are quantities that
cannot increase under the action of incoherent operations. Several such coherence measures
could be identified and include the relative entropy of coherence as well as the $l_1$-coherence
\cite{BaumgratzCP14}. While most definitions concerning the coherence power of operations
can be formulated for any choice of coherence measure, for explicit calculations it is of
advantage to consider the $l_1$-coherence measure
\begin{equation}
    C_{l_1}(\rho)=\sum_{i\neq j}|\rho_{ij}|.
\end{equation}

{\it Coherence properties of operations --}
Many physical questions relate to quantum operations and time evolution rather than
directly to quantum states. Hence it is of considerable interest to examine the
coherence properties of quantum operations or of their generators.
Let us begin with the
\begin{definition}\label{def1}
    The coherence power $P(\Phi)$ of a completely positive operation $\Phi$ is defined
    relative to the coherence measure $C(.)$ via
    \begin{equation}
        P(\Phi) = \max_{\rho} [C(\Phi(\rho)) - C(\rho)].
    \end{equation}
    For a unitary operation the coherence power is therefore
    \begin{equation}
        P(U) = \max_{\rho} [C(U\rho U^{\dagger}) - C(\rho)].
    \end{equation}
\end{definition}
We have deliberately left unrestricted the range over which the $\rho$ in the maximization
are taken. In \cite{ManiK2015} this range was restricted to the set of incoherent
states, i.e. the states for which $C(\rho)=0$. While this may appear to be a natural choice
it is not immediately clear that $C(\Phi(\rho)) - C(\rho)$ may actually be larger for some
$\rho$ with $C(\rho)>0$. Indeed, motivated by similar observations in the theory of entanglement
we consider this question and answer it in the affirmative~\cite{Note_added} in section \ref{Section4}.

Of interest in the context of dynamical systems are the time dependent generalizations
of the above concepts. Let us consider for example a time evolution $\Phi_t(\rho)$ with
generator ${\cal G}$, that is $\Phi_t = e^{{\cal G}t}$ or for the special case of a
unitary operator $U_t = e^{-iHt}$. Then one may either apply direction definition
\ref{def1} at a time $t$ or one may consider the coherence power of the generator by
\begin{definition} \label{defpower}
    For a time evolution $\Phi_t = e^{{\cal G}t}$ we determine the coherence power of the
    generator as
    \begin{equation}
        P({\cal G}) = \lim_{\Delta t\rightarrow 0} \frac{1}{\Delta t}\max_{\rho} [C(e^{{\cal G}\Delta t}\rho) - C(\rho)]
    \end{equation}
    and in case of unitary evolutions $U(t)=e^{-iHt}$ we write
    \begin{equation}
        P(H) = \lim_{\Delta t\rightarrow 0} \frac{1}{\Delta t}\max_{\rho} [C(e^{-iH\Delta t}\rho e^{iH\Delta t}) - C(\rho)].
    \end{equation}
\end{definition}
Note that one may also pursue questions concerning the coherence cost of an operation,
that is, the amount of coherence in the form of maximally coherent states that is required
to achieve an operation purely from incoherent operations. Questions regarding coherence cost
and distillable coherence have been addressed in \cite{WinterY2015}. We will not pursue such
quantities further here.

Of interest would be also to consider the N-dimensional unitary operations that have maximal coherence power. An example of this kind of unitaries would be the discrete Fourier transform:
\begin{corollary}
The coherence power of the discrete N-dimensional Fourier transform, calculated with respect to $l_1$-coherence, is maximal and is given by:
\begin{equation}
P_{l_1}(\mathcal{F})=N-1
\end{equation}
\end{corollary}
\textbf{Proof:}
\begin{eqnarray*}
    P_{l_1}(\mathcal{F}) &=& \\
    && \hspace*{-2.cm} =\max_\rho[\frac{1}{N}\sum_{a\ne b}|\sum_{j,j'}e^{\frac{2\pi i}{N}(ja-j'b)}\rho_{jj'}|-\sum_{a\ne b}|\rho_{ab}|]\\
    && \hspace*{-2.cm}  \geq\max_{\rho=\ket{k}\bra{k}}[\frac{1}{N}\sum_{a\ne b}|\sum_{j,j'}e^{\frac{2\pi i}{N}(ja-j'b)}\rho_{jj'}|-\sum_{a\ne b}|\rho_{ab}|]\\
    && \hspace*{-2.cm}  =N-1
\end{eqnarray*}
Since $P_{l_1}(U)\le N-1$, we conclude that a discrete N-dimensional Fourier transform is an example of unitary having maximal coherence power.

\section{Coherence power of a 2-dimensional unitary operator}\label{Section3}
As we have already mentioned, it is a non-trivial question whether it suffices in Definition
\ref{def1} to restrict $\rho$ to incoherent states or whether the full range of possible states,
including states with coherence, need to be considered. First we formulate and prove

\begin{theorem} The coherence power of a 2-dimensional unitary operation $U$ acting on qubits
and calculated with respect to the $l_1$-coherence is maximal for pure incoherent states $$P_{l_1}(U)=\max_{i=1,2}[C_{l_1}(U\ket{i}\bra{i} U^\dagger)].$$
\end{theorem}

\begin{proof}
First we note that the coherence power of $R_z(\alpha)UR_z(\beta)$ is the same as that for $U$.
\begin{eqnarray*}
    P(R_z(\alpha)UR_z(\beta)) &=& \\
    && \hspace*{-2.cm} =\max_\rho[C(R_z(\alpha)UR_z(\beta)\rho R_z^\dagger(\beta) U^\dagger R_z^\dagger(\alpha))-C(\rho)]\\
    && \hspace*{-2.cm}  =\max_\rho[C(UR_z(\beta)\rho R_z^\dagger(\beta) U^\dagger)-C(\rho)]\\
    && \hspace*{-2.cm}  =\max_\rho[C(U\rho U^\dagger)-C(R_z^\dagger(\beta)\rho R_z(\beta)]\\
    && \hspace*{-2.cm}  =\max_\rho[C(U\rho U^\dagger)-C(\rho)]\\
    && \hspace*{-2.cm} = P(U)
\end{eqnarray*}
Now consider
\begin{equation}
     M=\begin{pmatrix} e^{i(\psi+\alpha)}&0 \\ 0&e^{-i(\psi-\alpha)} \end{pmatrix}
     \begin{pmatrix} u_{gg} & u_{ge} \\ u_{eg} & u_{ee}  \end{pmatrix}
     \begin{pmatrix} e^{i(\phi+\beta)}&0 \\ 0&e^{-i(\phi-\beta)} \end{pmatrix}
     \nonumber
\end{equation}
where $\alpha$ and $\beta$ are global phases without physical effect. We choose $\alpha$
and $\psi$ such that $u_{gg}e^{i(\psi + \alpha)}\in\mathbb{R}^+ $ and
$u_{eg}e^{i(-\psi + \alpha)}\in\mathbb{R}$. Hence we find
\begin{equation}
     M = \begin{pmatrix} u_{gg}&u_{ge} \\ u_{eg}&u_{ee}  \end{pmatrix}\begin{pmatrix} e^{i\phi}&0 \\ 0&e^{-i\phi}  \end{pmatrix}\begin{pmatrix} e^{i\beta}&0 \\ 0&e^{i\beta}  \end{pmatrix}
\end{equation}
with $u_{gg}\in\mathbb{R}^+$ and $u_{eg}\in\mathbb{R}$. Now choose $\phi=-\beta$ and
make use of the orthonormality of the columns in a unitary
\begin{equation}
    u_{gg}(u_{ge}e^{-2i\phi}) + u_{eg}(u_{ee}e^{-2i\phi}) = 0
\end{equation}
to conclude from $u_{gg},u_{eg}\in\mathbb{R}$ that the phase of $u_{ge}e^{-2i\phi}$
and $u_{ee}e^{-2i\phi}$ is equal and can be eliminated by appropriate choice of $\phi$.
Hence we can assume
\begin{equation}
     M = \begin{pmatrix} u_{gg} & u_{ge} \\ u_{eg} & u_{ee}  \end{pmatrix}
\end{equation}
with $u_{gg},u_{eg},u_{ge} \mbox{and} u_{ee}\in \mathbb{R}$. Hence we can start by considering
real $U$ and using $\rho_{gg}=1-\rho_{ee}$ and $\rho_{eg}=\rho_{ge}e^{i\gamma}$ we find
\begin{eqnarray*}
    P(U) &=& 2\max_\rho[|u_{ee}u_{ge}+\rho_{gg}(u_{eg}u_{gg}-u_{ee}u_{ge})\\
    && \hspace{1.5cm} +\rho_{ge}(u_{ee}u_{gg}+e^{i\gamma}u_{eg}u_{ge})| -|\rho_{ge}|]
\end{eqnarray*}
As the first two terms are real and the third term can be chosen to have any phase
by virtue of the freedom of phase of $\rho_{ge}$ we notice that the absolute value
takes on its maximum value when $\rho_{ge}(u_{ee}u_{gg}+e^{i\gamma}u_{eg}u_{ge})$
is real and has the same sign as the sum of the first two terms.\\
Now let us choose $\rho_{ge}(u_{ee}u_{gg} + e^{i\gamma}u_{eg}u_{ge})\in\mathbb{R}$ and
with the same sign as $u_{ee}u_{ge}+\rho_{gg}(u_{eg}u_{gg}-u_{ee}u_{ge})$ (the case
for opposite sign is treated analogously). Then there are two cases:\\
1) $u_{ee}u_{ge}+\rho_{gg}(u_{eg}u_{gg}-u_{ee}u_{ge})>0$ which leads to
\begin{eqnarray*}
    P(U) &=& 2\max_\rho[(u_{ee}u_{ge}+\rho_{gg}(u_{eg}u_{gg}-u_{ee}u_{ge})\\
    && \hspace{1.5cm} +|\rho_{ge}|(|u_{ee}u_{gg}+e^{i\gamma}u_{eg}u_{ge}|-1)]
\end{eqnarray*}
As $U\in\mathbb{R}$ we have
\begin{displaymath}
    |u_{ee}u_{gg} + e^{i\gamma} u_{eg}u_{ge}| = \left|{\begin{pmatrix} u_{gg} \\ u_{ge}e^{i\gamma} \end{pmatrix}
    \begin{pmatrix} u_{ee} \\  u_{eg} \end{pmatrix}^{\dagger}}\right|
\end{displaymath}
As the vectors on the right are normalized the modulus of their scalar product is bounded by 1.
Therefore $2|\rho_{ge}|(|u_{ee}u_{gg}+e^{i\gamma}u_{eg}u_{ge}|-1) \le 0$ and takes its maximum
for $\rho_{eg}=0$.\\
2) $u_{ee}u_{ge}+\rho_{gg}(u_{eg}u_{gg}-u_{ee}u_{ge})<0$ proceeds along the same lines. \\
The coherence power of a 2-dimensional unitary is therefore achieved for states $\rho$ that
are incoherent. To complete the proof of the theorem we now note that by the convexity $C(\sum_n{p_n\rho_n})\le\sum_n{p_nC(\rho_n)}$ for any set of states $\{\rho_n\}$ and
probability distribution $\{p_n\}$ we find
\begin{eqnarray*}
    C_{l_1}(U\rho_{inc}U^\dagger) &=& C_{l_1}(U\sum_i{p_i\ket{i}\bra{i}}U^\dagger)\\
    &=& C_{l_1}(\sum_i{p_iU\ket{i}\bra{i}U^\dagger})\\
    &\le& \sum_i{p_iC_{l_1}(U\ket{i}\bra{i}U^\dagger})\\
    &\le& C_{l_1}(U\ket{i^*}\bra{i^*}U^\dagger)
\end{eqnarray*}
where $\ket{i^*}\bra{i^*}$ is the pure incoherent state which has the largest contribution
in the sum \cite{Footnote1}. This concludes the proof.
\end{proof}
From theorem 1 we easily find
\begin{corollary} The coherence power of a 2-dimensional unitary operation U, calculated
with respect to the $l_1$-coherence, is given by 
\begin{equation}
    P_{l_1}(U) = \max_j\{(\sum_{i=1}^{2}|U_{ij}|)^2:j=1,2\}-1
\end{equation}
\end{corollary}
\begin{proof}
Since in order to compute the coherence power of a 2-dimensional unitary we need to
maximize the gain over pure incoherent states only, we find
\begin{eqnarray*}
    P_{l_1}(U) &=& \max_{\ket{k}\bra{k}}[C_{l_1}(U\ket{k}\bra{k} U^\dagger)]:k=1,2\\
    &=& \max_j\{(\sum_{i=1}^{2}|U_{ij}|)^2:j=1,2\}-1
\end{eqnarray*}
\end{proof}

\section{Coherence power of an N-dimensional unitary operator ($N>2$)}\label{Section4}

Naively it might be expected that the coherence power of any quantum channel is
achieved on incoherent states. Indeed, the coherence power has been defined in
this way in \cite{ManiK2015}. However,
it is not self-evident that the largest coherence gain is obtained from incoherent
states. Indeed, in the theory of entanglement the analogous question, i.e. whether
the entanglement gain is maximized by starting on separable states, has been answered
in the negative \cite{BennettHL+2003,LeiferHL2003}. In the following we show that the
same observation holds for the case of coherence power.
\begin{proposition}
    For $N>2$, the coherence power of an N-dimensional unitary operator requires
    optimization over coherent states.
\end{proposition}
\textbf{Proof:} We consider the coherence power as quantified relative to the
$l_1$-coherence and the relative entropy of coherence \cite{BaumgratzCP14}.

{\it $l_1$-coherence power --} Let us consider a 3-dimensional rotation
by $\theta=\frac{\pi}{4}$ around the \textit{x} axis:
\begin{equation}
    R_x\left(\frac{\pi}{4}\right) =
    \begin{pmatrix}
        1 &         0         &           0 \\
        0 & \frac{1}{\sqrt 2} & -\frac{1}{\sqrt 2} \\
        0 & \frac{1}{\sqrt 2} & \frac{1}{\sqrt 2}
   \end{pmatrix}
\end{equation}
According to Corollary 1, its maximum coherence gain calculated over pure incoherent
states is found to be:
\begin{equation}
    \max_j\{(\sum_{i=1}^{3}|R_x\left(\frac{\pi}{4}\right)_{ij}|)^2:j=1,2,3\}-1 = 1.
\end{equation}
It is easy to find examples of coherent states that provide a larger coherence
gain for this particular rotation. The state $\ket{\psi}=c_1\ket{1}+c_3\ket{3}$
where $c_1=0.3$ and $c_3=\sqrt{1-0.3^2}$, for instance, provides a coherence gain
of $1.1471$:
\begin{eqnarray*}
    G_{\ket{\psi}\bra{\psi}}\left(R_x\left(\frac{\pi}{4}\right)\right) &=& \\
    && \hspace*{-2.cm} =C_{l_1}\begin{pmatrix}{c_1}^{2} & −\frac{c_1\,c_3}{\sqrt{2}} & \frac{c_1\,c_3}{\sqrt{2}}\\ −\frac{c_1\,c_3}{\sqrt{2}} & \frac{{c_3}^{2}}{2} & −\frac{{c_3}^{2}}{2}\\ \frac{c_1\,c_3}{\sqrt{2}} & −\frac{{c_3}^{2}}{2} & \frac{{c_3}^{2}}{2}\end{pmatrix}\ - C_{l_1}\begin{pmatrix}{c_1}^{2} & 0 & c_1\,c_3 \\ 0 & 0 & 0 \\ c_1\,c_3 & 0 & {c_3}^{2}\end{pmatrix}\\
    && \hspace*{-2.cm}  =(2\sqrt 2-2)c_1c_3+c_3^2\\
    && \hspace*{-2.cm} =1.1471>1.
\end{eqnarray*}

{\it Relative entropy of coherence power --} Assuming that the coherence power could
be calculated by maximization of the gain over incoherent states, and the observation
that by convexity of the relative entropy of coherence we can then restrict maximization
to pure incoherent states, we find for the coherence power of an N-dimensional unitary
with respect to the relative entropy of coherence:
\begin{equation}
    P_{rel.ent.}(U) = \max_i\{-\sum_{j=1}^N|U_{ij}|^2\log(|U_{ij}|^2):i=1,...,N\}
\end{equation}
\textbf{Proof:}
\begin{eqnarray*}
    P_{rel.ent.}(U) &=& \\
    && \hspace*{-2.cm} = \max_{\ket{i}\bra{i}}[ C_{rel.ent}(U\ket{i}\bra{i} U^\dagger)-C_{rel.ent.}(\ket{i}\bra{i}):i=1,...,N]\\
    && \hspace*{-2.cm} = \max_{\ket{i}\bra{i}}[ S((U\ket{i}\bra{i}U^\dagger)_{\diag})
    - S(U\ket{i}\bra{i} U^\dagger):i=1,...,N]\\
    && \hspace*{-2.cm} = \max_{\ket{i}\bra{i}}[S((U\ket{i}\bra{i} U^\dagger)_{\diag}):i=1,...,N]\\
    && \hspace*{-2.cm} = \max_i[-\sum_{j=1}^N|U_{ij}|^2\log(|U_{ij}|^2):i=1,...,N].
\end{eqnarray*}
Let us now consider a 3-dimensional rotation of $\theta=\frac{\pi}{8}$ around the $x$ axis:
\begin{equation}
    R_x\left(\frac{\pi}{8}\right) =
    \begin{pmatrix}
        1 &              0                 &                 0 \\
        0 & \cos\left(\frac{\pi}{8}\right) & -\sin\left(\frac{\pi}{8}\right)\\
        0 & \sin\left(\frac{\pi}{8}\right) & \cos\left(\frac{\pi}{8}\right).
   \end{pmatrix}
\end{equation}
Maximization of the coherence gain of this rotation over incoherent states results
in
\begin{eqnarray*}
    \max_i\{-\sum_{j=1}^3|R_x\left(\frac{\pi}{8}\right)_{ij}|^2
    \log(|R_x\left(\frac{\pi}{8}\right)_{ij}|^2):i=1,2,3\}&& \\
    && \hspace*{-1.cm}=0.41650.
\end{eqnarray*}
However we have found a number of coherent states that provide an even larger gain, such as the state $\ket{\phi}=q_2\ket{2}+q_3\ket{3}$ where $q_2=\sqrt{1-0.12533^2}$ and $q_3=0.12533$:
\begin{equation}
G_{\ket{\phi}\bra{\phi}}\left(R_x\left(\frac{\pi}{8}\right)\right)=0.47648>0.41650.
\end{equation}

The maximum gain of these two rotations, with respect to their corresponding coherence
measure, is not achieved on pure incoherent states.
Therefore the most natural definition of the coherence coherence power is by maximization
over {\it all} states.

\section{Conclusion}
In this note we have discussed several possible definitions of coherence power. We have 
also proved that the coherence power of a 2-dimensional unitary operator can be calculated 
by maximizing its coherence gain over pure incoherent states only. Giving two explicit 
counterexamples, we could show that this result cannot be generalized for dimensions higher 
than $N=2$ \cite{Note_added}.\\

Hence, analogously to the result of entanglement theory, where it was observed that entangled 
states typically admit the largest gain in entanglement, we found that some initial coherence 
in the input state can be required for an optimal coherence gain to be attained. This result 
shows that it is not sufficient to maximize the coherence gain over incoherent states. It 
seems therefore an interesting question if one can restrict the optimization in higher dimension 
to a smaller subset or one needs to run it over the whole state space even for unitary evolutions. 
For non-unitary evolutions, while it seems challenging to try to find a generic simplification, 
one still might use the symmetries present in coherence theory to simplify the optimization for 
a given evolution, similarly as we used them here in the case of qubits and unitary evolution 
for proving theorem 1.

\begin{acknowledgements}
We acknowledge discussions with S.F. Huelga, N. Killoran, A. Smirne, M. Matera and K. Macieszczak. This work was
supported by an Alexander von Humboldt Professorship, the EU Integrating Project SIQS
as well as the EU STREPs EQUAM and QUCHIP.
\end{acknowledgements}


\begin{thebibliography}{99}
\bibitem{Footnote1} Note that this has also been observed in \cite{ManiK2015}, where the coherence
power of an arbitrary quantum channel calculated for any coherence measure is proved to be
maximum for pure states, provided that the maximization is performed over incoherent states
only.
\bibitem{Note_added}  After completion of this work we became aware of \cite{BuZW2015}
which independently found that the coherence power is generally achieved
only by maximization over the full state space and also has a different proof for theorem 1 and corollary 2.
\end{thebibliography}
\end{document}